\documentclass[a4paper,11pt]{article}

\usepackage{amsthm}
\usepackage{fullpage}%
\usepackage{graphicx,tikz,enumerate}
\usepackage[colorlinks=true]{hyperref}
\usepackage[title]{appendix}
\usepackage{amssymb, amsfonts,mathtools,stmaryrd}
\usepackage{xcolor}
\usepackage{enumitem}
\usepackage{etoolbox}

\usepackage{algorithm}
\usepackage{algpseudocode}

\usepackage{dsfont}
\usepackage{bbm}
\newtheorem{definition}{\textbf{Definition}}

\newtheorem{theorem}[definition]{\textbf{Theorem}}

\newtheorem{lemma}[definition]{\textbf{Lemma}}

\newtheorem{corollary}[definition]{\textbf{Corollary}}

\newtheorem{proposition}[definition]{\textbf{Proposition}}

\newtheorem{example}[definition]{\textbf{Example}}

\newtheorem*{proof*}{\textbf{Proof}}

\newcommand{\N}{\mathbb{N}}

\newcommand{\msf}[1]{\mathsf{#1}}
\newcommand{\Idm}{\{1,\ldots,m\}}
\newcommand{\I}{\mathsf{I}}

\newcommand{\LTL}{\textrm{LTL}}

\newcommand{\CTL}{\textrm{CTL}}

\newcommand{\post}{\mathrm{post}}
\newcommand{\paths}{\mathrm{paths}}

\newcommand{\prop}{\mathsf{Prop}}

\DeclareMathOperator{\lF}{\mathbf{F}}
\DeclareMathOperator{\lG}{\mathbf{G}}
\DeclareMathOperator{\lU}{\mathbf{U}}
\DeclareMathOperator{\lX}{\mathbf{X}}
\DeclareMathOperator{\lR}{\mathbf{R}}
\DeclareMathOperator{\lW}{\mathbf{W}}
\DeclareMathOperator{\lM}{\mathbf{M}}

\DeclareMathOperator{\lA}{\mathbf{A}}
\DeclareMathOperator{\lE}{\mathbf{E}}

\usepackage{pifont}
\usepackage{authblk}
\usepackage{graphicx}
\author[1,2]{Benjamin Bordais}
\author[1,2]{Daniel Neider}
\author[3]{Rajarshi Roy}
\affil[1]{TU Dortmund University, Dortmund, Germany}
\affil[2]{Center for Trustworthy Data Science and Security, University Alliance Ruhr, Dortmund, Germany}
\affil[3]{Max Planck Institute for Software Systems, Kaiserslautern, Germany}

\date{}

\begin{document}
	\title{Learning Temporal Properties is $\msf{NP}$-hard}
	\maketitle	
	
	\section{Introduction}
	
	Temporal logics are probably the most popular specification languages employed to specify the temporal behavior of computing systems.
	Originally introduced in the context of program verification, temporal logics such as Linear Temporal Logic (LTL)~\cite{pnueli77} and Computation Tree Logic (CTL)~\cite{ClarkeE81} now find applications in a wide variety of topics, including reinforcement learning~\cite{FuT14,SadighKCSS14,LiVB17,CamachoIKVM19}, motion planning~\cite{GiacomoV99,FainekosKP05}, formal verification~\cite{pnueli77,ClarkeE81,QueilleS82,CimattiCGR00,Biere09}, reactive synthesis~\cite{PnueliR89,ScheweF07a}, process mining~\cite{CecconiGCMM20,CecconiGCMM22}, and so on.
	The primary reason for the popularity of temporal logics is that its specifications are not only mathematically rigorous but also human-understandable due to the resemblance to natural language.
	
	Most applications that employ temporal logics, even today, rely on manually writing specifications based on human intuition.
	This manual process, however, is inefficient and tedious and often leads to several errors in formulating specifications~\cite{AmmonsBL02,Rozier16,BjornerH14}.
	To mitigate these issues, recent research has focused on automatically learning temporal logic formulae based on observations of the underlying system.
	
	Perhaps, the most basic setting for learning temporal logic is that of inferring a concise formula from a given set of positive (or desirable) and negative (or undesirable) executions of a system.
	The focus on learning a \emph{concise} formula is inspired by the principle of Occam's razor, which advocates that formulae that are small in size are less complex and thus easy to understand.
	For this basic setting, several works~\cite{flie,CamachoM19,scarlet} have devised efficient algorithms with impressive empirical performance.
	
	However, the theoretical analysis of the problem of learning temporal logic in terms of computational complexity is significantly lacking.
	The only notable work that deals with the complexity of learning temporal logic is the work by Fijalkow and Lagarde~\cite{FijalkowL21}.
	They show that learning certain fragments of $\LTL$ (i.e. subclasses of $\LTL$ with a restricted set of temporal operators) from finite executions is $\msf{NP}$-complete.
	
	In this work, we conduct a detailed study of the computational complexity of learning temporal properties in its full generality.
	We assume that we are given finite sets $\mathcal{P}$ and $\mathcal{N}$ of positive and negative examples of system behavior, respectively, and a desired size $B$ given in unary for the prospective formula as input.
	We now study the following decision problem: given $\mathcal{P}$, $\mathcal{N}$ and $B$, does there exist a formula $\varphi$ of size less than or equal to $B$ that all the positive examples satisfy and that all negative examples violate. Note that, contrary to what is done in~\cite{FijalkowL21} by Fijalkow, the bound $B$ that we consider is given in unary, not in binary. Then, we consider the above problem (i) for $\LTL$-formulae, in which case, we consider the examples to be formalized as \emph{infinite words}, and (ii) for $\CTL$-formulae, in which case, we consider the examples to be system models formalized as \emph{Kripke structures}. We show that both learning problems, for $\LTL$ and $\CTL$ are $\msf{NP}$-hard. Furthermore, it is rather straightforward to show that $\LTL$ and $\CTL$ learning is in $\msf{NP}$. Indeed, one can simply guess a solution, i.e., a prospective formula, of size less than or equal to $B$ --- which is possible since $B$ is written in unary --- and verify that the solution guessed is correct in polynomial time. Note that this is possible for both $\LTL$ and $\CTL$ formulae~\cite{MarkeyS03}.
	
	Our proof of $\msf{NP}$-hardness for $\LTL$ learning is done via a reduction from the satisfiability problem ($\msf{SAT}$). The reduction itself is easy to define, the main difficulty lies in proving that from an $\LTL$-formula separating the words in $\mathcal{P}$ and $\mathcal{N}$, one can extract a satisfying valuation (for the $\msf{SAT}$ formula that we consider). The proof of this fact follows two consecutive stages: first, we show that it is sufficient to only consider the $\LTL$-formulae satisfying a specific property, namely being \emph{concise}, see Definition~\ref{def:concise_formulae}. Second, we show that these concise formulae imply --- in a specific sense --- an $\LTL$-formula that is \emph{disjunctive}, see Definition~\ref{def:disjunctive}, that is itself an $\LTL$-formula separating the words in $\mathcal{P}$ and $\mathcal{N}$. From such a disjunctive formula, it is then straightforward to extract a satisfying valuation. The proof of $\msf{NP}$-hardness for $\CTL$ learning is straightforward via a reduction from the $\LTL$ learning case.
	
	\section{$\LTL$ and $\CTL$ learning}
	
	\subsection{$\LTL$: syntax and semantics}
	Before introducing $\LTL$-formulae \cite{DBLP:conf/focs/Pnueli77}, let us first introduce the objects on which these formulae will be interpreted: infinite words. 
	
	\paragraph{Infinite and ultimately periodic words.}
	Given a set of propositions $\prop$, an infinite word is an element of the set $(2^\prop)^\omega$. For all infinite words $w \in (2^\prop)^\omega$, for all $i \in \N$, we let $w[i] \in 2^\prop$ denote the element in $2^\prop$ at position $i$ of the word $w$ and $w[i:] \in (2^\prop)^\omega$ denote the infinite word starting at position $i$.
	
	Furthermore, we will be particularly interested in ultimately periodic words, and in size-1 ultimately periodic words. They are formally defined below.
	\begin{definition}
		Consider a set of propositions $\prop$. An ultimately periodic word $w \in (2^\prop)^\omega$ is such that $w = u \cdot v^\omega$ for some finite words $u \in (2^\prop)^*,v \in (2^\prop)^+$. In that case, we set the size $|w|$ of $w$ to be equal to $|w| := |u| + |v|$. Then, for all sets $S$ of ultimately periodic words, we set $|S| := \sum_{w \in S} |w|$.
		
		An ultimately periodic word $w$ is of size-1 if $|w| = 1$. That is, $w = \alpha^\omega$ for some $\alpha \subseteq \prop$. 
	\end{definition}

	\paragraph{Syntax.}
	Given a set of propositional variables $\prop$, $\LTL$-formulae are defined inductively as follows:
	\begin{align*}
		\varphi \Coloneqq p & \mid \neg \varphi \mid \varphi \wedge \varphi \mid \varphi \vee \varphi \mid \varphi \Rightarrow \varphi \mid \varphi \Leftrightarrow \varphi \\
		& \mid \lX \varphi \mid \lF \varphi \mid \lG \varphi \mid \varphi \lU \varphi \mid \varphi \lR \varphi \mid \varphi \lW \varphi \mid \varphi \lM \varphi
	\end{align*}
	where $p\in \prop$. These formulae use the following temporal operators:
	\begin{itemize}
		\item $\lX$: neXt;
		\item $\lF$: Future;
		\item $\lG$: Globally;
		\item $\lU$: Until;
		\item $\lR$: Release;
		\item $\lW$: Weak until;
		\item $\lM$: Mighty release.
	\end{itemize}
	 We let $\mathsf{Op}_{Un} := \mathsf{Op}_{Un}^{\msf{lg}} \cup \mathsf{Op}_{Un}^{\msf{temp}}$ where $\mathsf{Op}_{Un}^{\msf{lg}} := \{ \neg \}$ and $\mathsf{Op}_{Un}^{\msf{temp}} := \{ \lX,\lF,\lG\}$ denote the set of unary operators. We also let $\mathsf{Op}_{Bin} := \mathsf{Op}_{Bin}^{\msf{lg}} \cup \mathsf{Op}_{Bin}^{\msf{temp}}$ where $\mathsf{Op}_{Bin}^{\msf{lg}} := \{ \wedge,\vee,\Rightarrow,\Leftrightarrow \}$ and $\mathsf{Op}_{Bin}^{\msf{temp}} := \{\lU,\lR,\lW,\lM\}$ denote the set of binary operators. Note that these operators will also be used in $\CTL$ formulae. We also let $\LTL$ denote the set of all $\LTL$-formulae. We then define the size $\msf{Sz}(\varphi)$ of an $\LTL$-formula $\varphi \in \LTL$ to be its number of sub-formulae. That set of sub-formulae $\msf{SubF}(\varphi)$ is defined inductively as follows:
	\begin{itemize}
		\item $\msf{SubF}(p) := \{ p \}$ for all $p \in \prop$;
		\item $\msf{SubF}(\bullet \varphi) := \msf{SubF}(\varphi) \cup \{ \bullet \varphi \}$ for all unary operators $\bullet \in \mathsf{Op}_{Un}$;
		\item $\msf{SubF}(\varphi_1 \bullet \varphi_2) := \msf{SubF}(\varphi_1) \cup  \msf{SubF}(\varphi_2) \cup \{ \varphi_1 \bullet \varphi_2 \}$ for all binary operators $\bullet \in  \mathsf{Op}_{Bin}$.
	\end{itemize}
	
	\paragraph{Semantics.}
		 The semantics of $\LTL$-formulae is defined inductively as follows, for all $w \in (2^\prop)^\omega$:
	\begin{align*}
		w \models p & \text{ iif } p \in w[0], \\
		w \models \neg \varphi & \text{ iif } w \not\models \varphi; \\
		w \models \varphi_1 \wedge \varphi_2 & \text{ iif } w \models \varphi_1 \text{ and } w \models \varphi_2; \\
		w \models \varphi_1 \vee \varphi_2 & \text{ iif } w \models \varphi_1 \text{ or } w \models \varphi_2; \\
		w \models \varphi_1 \Rightarrow \varphi_2 & \text{ iif } w \models \neg \varphi_1 \vee \varphi_2; \\
		w \models \varphi_1 \Leftrightarrow \varphi_2 & \text{ iif } w \models (\varphi_1 \wedge \varphi_2) \vee (\neg \varphi_1 \wedge \neg \varphi_2); \\
		w \models \lX \varphi & \text{ iif } w[1:] \models \varphi; \\
		w \models \lF \varphi & \text{ iif } \exists i \in \N,\; w[i:] \models \varphi; \\
		w \models \lG \varphi & \text{ iif } \forall i \in \N,\; w[i:] \models \varphi; \\
		w \models \varphi_1 \lU \varphi_2 & \text{ iif } \exists i \in \N,\; w[i:] \models \varphi_2 \text{ and } \forall j \leq i-1,\; w[j:] \models \varphi_1; \\
		w \models \varphi_1 \lR \varphi_2 & \text{ iif } w \models \neg (\neg \varphi_1 \lU \neg \varphi_2) \\
		w \models \varphi_1 \lW \varphi_2 & \text{ iif } w \models (\varphi_1 \lU \varphi_2) \vee \lG\varphi_1; \\
		w \models \varphi_1 \lM \varphi_2 & \text{ iif } w \models (\varphi_1 \lR \varphi_2) \wedge \lF \varphi_1
	\end{align*}
	
	When $w \models \varphi$, we say that $\varphi$ accepts $w$, otherwise it rejects it. In the following, we will consider ultimately periodic words, i.e. words of the shape $u \cdot v^\omega$. 
	
	\subsection{$\CTL$: syntax and semantics}
	Before introducing $\CTL$-formulae, let us first introduce Kripke structures, i.e. the objects on which $\CTL$-formulae will be interpreted. The way the following definitions are presented is essentially borrowed from \cite{DBLP:journals/corr/abs-2310-13778}.
	
	\paragraph{Kripke structure.}
	A \textit{Kripke structure} $M = (S,I,R,L)$ over a finite set $\prop$ of propositions consists of a set of states $S$, a set of initial states $I\subseteq S$, a transition relation $R\subseteq S\times S$ such that for all states $s \in S$ there exists a state $s' \in S$ satisfying $(s,s')\in R$, and a labeling function $L\colon S\rightarrow 2^{\prop}$. We define the size of $M$ as $\lvert S\rvert$. For all states $s \in S$, the set~$\post(s) = \{s'\in S \mid (s,s') \in R\}$ contains all successors of $s \in S$.
	
	A (infinite) \emph{path} in the Kripke structure $M$ is an infinite sequence of states $\pi = s_0s_1\dots$ such that $s_{i+1} \in \post(s_i)$ for each $i\geq 0$. For a state $s$, we let $\paths(s)$ denote the set of all (infinite) paths starting from $s$.
	
	\paragraph{Syntax.}
	To define the syntax of $\CTL$-formulae, we introduce two types of formulae: state formulae and path formulae. Intuitively, state formulae express properties of states, whereas path formulae express temporal properties of paths. 
	For ease of notation, we denote state formulae and path formulae with the Greek letter $\varphi$ and the Greek letter $\psi$, respectively.	$\CTL$-state formulae over the set of propositions $\prop$ are given by the grammar:
	\begin{align*}
		\varphi \Coloneqq p \mid \neg \varphi \mid \varphi \wedge \varphi \mid \varphi \vee \varphi \mid \varphi \Rightarrow \varphi \mid \varphi \Leftrightarrow \varphi \mid \lE \psi \mid \lA\psi 
	\end{align*}
	where $p\in \prop$, and $\psi$ is a path formula. Next, $\CTL$-path formulae are given by the grammar
	\begin{align*}
		\psi \Coloneqq \lX \varphi \mid \lF \varphi \mid \lG \varphi \mid \varphi \lU \varphi \mid \varphi \lR \varphi \mid \varphi \lW \varphi \mid \varphi \lM \varphi
	\end{align*}
	When we refer to $\CTL$-formulae, we usually refer to $\CTL$-path formulae. We denote by $\CTL$ the set of all 
	$\CTL$-formulae. The set of sub-formulae $\msf{SubF}(\Phi)$ of a $\CTL$-formula $\varphi$ is then defined inductively as follows:
	\begin{itemize}
		\item $\msf{SubF}(p) := \{ p \}$ for all $p \in \prop$;
		\item $\msf{SubF}(\neg \Phi) := \{ \neg \Phi \} \cup \msf{SubF}(\Phi)$;
		\item $\msf{SubF}(\diamond (\bullet \Phi)) := \{ \diamond (\bullet \Phi) \} \cup \msf{SubF}(\Phi)$ for all $\bullet \in \mathsf{Op}_{Un}^{\msf{temp}}$ and $\diamond\in \{ \exists,\forall \}$;
		\item $\msf{SubF}(\Phi_1 \bullet \Phi_2) := \{ \Phi_1 \bullet \Phi_2 \} \cup \msf{SubF}(\Phi_1) \cup \msf{SubF}(\Phi_2)$ for all $\bullet \in \mathsf{Op}_{Bin}^{\msf{lg}}$;
		\item $\msf{SubF}(\diamond (\Phi_1 \bullet \Phi_2)) := \{ \diamond (\Phi_1 \bullet \Phi_2) \} \cup \msf{SubF}(\Phi_1) \cup \msf{SubF}(\Phi_2)$ for all $\bullet \in \mathsf{Op}_{Bin}^{\msf{temp}}$  and $\diamond\in \{ \exists,\forall \}$.
	\end{itemize}
	
	The size $|\varphi|$ of a $\CTL$-formula is defined as its number of sub-formulae: $|\varphi| := |\msf{SubF}(\varphi)|$. 
	
	\paragraph{Semantics.}
	As mentioned above, we interpret 
	CTL-formulae over Kripke structures using the standard definitions~\cite{model-checking-book}.	Given a state $s$ and state formula~$\Phi$, we define when $\varphi$ holds in state $s$, denoted using $s\models \varphi$, inductively as follows:
	\begin{align*}
		s \models p & \text{ iif } p \in L(s), \\
		s \models \neg \varphi & \text{ iif } s \not\models \varphi, \\
		s\models \varphi_1\wedge\varphi_2 & \text{ iif } s\models\varphi_1 \text{ and } s\models\varphi_2, \\
		s\models \varphi_1\lor\varphi_2 & \text{ iif } s\models\varphi_1 \text{ or } s\models\varphi_2, \\
		s\models \varphi_1\Rightarrow\varphi_2 & \text{ iif } s\models \neg\varphi_1 \lor \varphi_2, \\
		s\models \varphi_1\Leftrightarrow\varphi_2 & \text{ iif } s\models (\varphi_1 \wedge \varphi_2) \lor (\neg \varphi_1 \wedge \neg \varphi_2), \\
		s\models \lE\psi & \text{ iif } \pi\models\psi \text{ for some } \pi\in\paths(s) \\
		s\models \lA\psi & \text{ iif } \pi\models\psi \text{ for all } \pi\in\paths(s)
	\end{align*}
	Furthermore, given a path $\pi$ and a path formula $\psi$, we define when $\psi$ holds for the path $\pi$, also denoted using $\pi\models \varphi$, inductively as follows:
	\begin{align*}
		\pi \models \lX \varphi & \text{ iif } \pi[1:] \models \varphi; \\
		\pi \models \lF \varphi & \text{ iif } \exists i \in \N,\; \pi[i:] \models \varphi; \\
		\pi \models \lG \varphi & \text{ iif } \forall i \in \N,\; \pi[i:] \models \varphi; \\
		\pi \models \varphi_1 \lU \varphi_2 & \text{ iif } \exists i \in \N,\; \pi[i:] \models \varphi_2 \text{ and } \forall j \leq i-1,\; \pi[j:] \models \varphi_1; \\
		\pi \models \varphi_1 \lR \varphi_2 & \text{ iif } \pi \models \neg (\neg \varphi_1 \lU \neg \varphi_2) \\
		\pi \models \varphi_1 \lW \varphi_2 & \text{ iif } \pi \models (\varphi_1 \lU \varphi_2) \vee \lG\varphi_1; \\
		\pi \models \varphi_1 \lM \varphi_2 & \text{ iif } \pi \models (\varphi_1 \lR \varphi_2) \wedge \lF \varphi_1
	\end{align*}
	We now say that a $\CTL$-formula $\varphi$ holds on a Kripke structure $M$, denoted by $M\models \varphi$, if $s\models\varphi$ for all initial states $s\in I$ of $M$.
	
	\subsection{Decision problems}
	We define the $\LTL$ and $\CTL$ learning problems below, where an example for $\LTL$ is an ultimately periodic word, whereas an example for $\CTL$ is a Kripke structure.
	\begin{definition}
		Let $\msf{L} \in \{ \LTL,\CTL \}$. We denote by $\msf{L}_\msf{Learn}$ the following decision problem:
		\begin{itemize}
			\item Input: $(\prop,\mathcal{P},\mathcal{N},B)$ where $\prop$ is a set of propositions, $\mathcal{P},\mathcal{N}$ are two finite sets of examples for $\LTL$, and $B \in \N$.
			\item Output: yes iff there exists an $\LTL$-formula $\varphi \in \LTL$ such that $|\varphi| \leq B$ that is separating $\mathcal{P}$ and $\mathcal{N}$, i.e. such that:
			\begin{itemize}
				\item for all $X \in \mathcal{P}$, we have $X \models \varphi$;
				\item for all $X \in \mathcal{N}$, we have $X \not\models \varphi$.
			\end{itemize}
		\end{itemize}
		The size of the input is equal to $|\prop| + |\mathcal{P}| + |\mathcal{N}| + B$ (i.e. $B$ is written in unary).
	\end{definition}

	As mentioned in the introduction, since the model checking problems for $\LTL$ and $\CTL$ can be solved in polynomial time, it follows that both problems $\LTL_\msf{Learn}$ and $\CTL_\msf{Learn}$ are in $\msf{NP}$, as stated below.
	
	\begin{proposition}
		\label{prop:easy}
		The decision problem $\LTL_\msf{Learn}$ (resp. $\CTL_\msf{Learn}$) is in $\msf{NP}$.
	\end{proposition}
	
	\section{$\LTL$ learning is $\msf{NP}$-hard}
	The goal of this section is to show the theorem below:
	\begin{theorem}
		\label{thm}
		The decision problem $\LTL_\msf{Learn}$ is $\msf{NP}$-hard.
	\end{theorem}

	Theorem~\ref{thm} above will be proved via a reduction from $3-\msf{SAT}$. Hence, for the remaining of this section, we are going to consider a $\msf{SAT}$-formula:
	\begin{equation*}
		\Phi := \bigwedge_{i = 1}^n (l_1^i \vee l^i_2 \vee l^i_3) 
	\end{equation*}
	on a set of variable $X := \{ x_k \mid 1 \leq k \leq m \}$ where, for all $1 \leq i \leq n$ and $1 \leq j \leq 3$, we have $l_j^i$ a literal equal to either $l_j^i = x$ or $l_j^i = \lnot x$ for some $x \in X$. Then, $\Phi$ is a positive instance of the problem $3-\msf{SAT}$ if there is a valuation $v: X \rightarrow \{ \top,\bot \}$ satisfying the formula $\Phi$, denoted $v \models \Phi$.
	
	
	Let us formally define below the reduction that we will consider.
	\begin{definition}
		\label{def:reduction}
		We let:
		\begin{itemize}
			\item $\prop := \{ x_k,\bar{x_k} \mid 1 \leq k \leq m \}$ denote the set of propositions that we will consider;
			\item for all $1 \leq i \leq n$: $C_i := \{ \tilde{l_1^i},\tilde{l_2^i},\tilde{l_3^i} \} \in (2^\prop)$ with, for all $x \in X$ and $1 \leq j \leq 3$:
			\begin{align*}
			\tilde{l_j^i} := \begin{cases}
			x & \text{ if }l_j^i = x \\
			\bar{x} & \text{ if }l_j^i = \neg x \\
			\end{cases}
			\end{align*}
		\end{itemize}
		Then, we define the input $\msf{In}_{\Phi,X} = (\prop,\mathcal{P},\mathcal{N},B)$ of the $\LTL_\msf{Learn}$ decision problem where:
		\begin{itemize}
			\item $\mathcal{P} := \{ c_i \mid 1 \leq i \leq n \} \uplus \{ \rho_k \mid 1 \leq k \leq m \}$ where for all $1 \leq i \leq n$, we have $c_i := C_i^\omega \in (2^\prop)^\omega$ and for all $1 \leq k \leq m$, $\rho_k := (\{ x_k,\bar{x_k} \})^\omega \in (2^\prop)^\omega$;
			\item $\mathcal{N} := \{ \eta \}$ with $\eta := \emptyset^\omega \in (2^\prop)^\omega$;
			\item $B = 2 m - 1$.
		\end{itemize}
	\end{definition}

	Let us describe on an example below what this reduction amounts to on a specific $\msf{SAT}$ formula $\Phi$. Note that we will come back twice to this example in the following.
	\begin{example}
		\label{ex:1}
		Assume that $X = \{ x_1,x_2,x_3 \}$ and:
		\begin{equation*}
			\Phi := x_1 \wedge (x_2 \lor \lnot x_3 \lor \lnot x_1) \wedge (\lnot x_2 \lor x_3)
		\end{equation*}
		Then, we have $\prop = \{ x_1,\bar{x_1},x_2,\bar{x_2},x_3,\bar{x_3}\}$, $\rho_1 = (\{ x_1,\bar{x_1} \})^\omega$, $\rho_2 = (\{ x_2,\bar{x_2} \})^\omega$ and $\rho_3 = (\{ x_3,\bar{x_3} \})^\omega$. Furthermore, $C_1 = \{ x_1 \}$, $C_2 = \{ \bar{x_1},x_2,\bar{x_3} \}$ and $C_3 = \{ \bar{x_2},x_3 \}$. Finally, $B = 5$.
	\end{example}

	Clearly, $\msf{In}_{\Phi,X}$ can be computed in polynomial time, given $(\Phi,X)$. The goal now is to show that $\Phi$ is satisfiable if and only if $\msf{In}_{\Phi,X}$ is a positive instance of the $\LTL_\msf{Learn}$ decision problem. 
		
	\subsection{The easy direction}
	We start with the easy direction, i.e. we show that from a satisfying valuation $v$, we can define an $\LTL$-formula of size (at most) $B$ that accepts all words in $\mathcal{P}$ and rejects all words in $\mathcal{N}$. This is formally stated and proved below.
	\begin{lemma}
		\label{lem:easy}
		If $\Phi$ is satisfiable then $\msf{In}_{\Phi,X}$ is a positive instance of the $\LTL_\msf{Learn}$ decision problem.
	\end{lemma}
	\begin{proof}
		Consider a valuation $v: X \rightarrow \{ \top,\bot \}$ that satisfies the formula $\Phi$. For all $1 \leq k \leq m$, we let $x_k^v \in \prop$ be such that:
		\begin{align*}
		x_k^v := \begin{cases}
		x_k & \text{ if }v(x_k) = \top \\
		\bar{x_k} & \text{ if }v(x_k) = \bot
		\end{cases}
		\end{align*}
		Then, we consider the $\LTL$-formula:
		\begin{equation}
		\varphi := x_1^v \lor x_2^v \lor \ldots \lor x_m^v
		\end{equation} 
		Clearly, we have $\msf{Sz}(\varphi) = 2m-1 = B$. Furthermore, consider any $1 \leq k \leq m$. We have $\rho_k \models x_k^v$, hence $\rho_k \models \varphi$. Furthermore, $\eta \not\models \varphi$. Consider now some $1 \leq i \leq n$. Since $v$ satisfies $\Phi$, there is some $1 \leq j \leq 3$ and $1 \leq k \leq m$ such that $l_j^i = x_k$ and $v(x_k) = \top$ or $l_j^i = \neg x_k$ and $v(x_k) = \bot$. In the first case, we have $x_k^v = x_k = \tilde{l_j^i} \in C_i$. In the second case, we have $x_k^v = \bar{x_k} = \tilde{l_j^i} \in C_i$. Thus, in any case, we have $c_i  = C_i^\omega \models \varphi$. 
		
		Therefore, $\varphi$ accepts all words in $\mathcal{P}$ and rejects all words in $\mathcal{N}$.
	\end{proof}
	
	Note that there is not only one such $\LTL$-formula, we give an example below of another that would do the job.
	\begin{example}
		Continuing Example~\ref{ex:1}, one can realize that the valuation $v: X \rightarrow \{ \top,\bot \}$ such that $v(x_1) = v(x_2) = v(x_3) = \top$ satisfies the formula $\Phi$. 
		
		Hence, the $\LTL$-formula $\varphi$ corresponding to this valuation that is defined in the above proof is:
		\begin{equation*}
			\varphi := x_1 \lor x_2 \lor x_3
		\end{equation*}
		However, other $\LTL$-formulae $\varphi'$ could have worked as well, for instance: 
		\begin{equation*}
			\varphi' := (x_1 \Leftrightarrow x_2) \Rightarrow x_3
		\end{equation*}
	\end{example}
	
	\subsection{The hard direction}
	There now remains the hard direction, i.e. showing that from an $\LTL$-formula of size (at most) $B$ that accepts all words in $\mathcal{P}$ and rejects all words in $\mathcal{N}$, it is possible to define a satisfying valuation $v$, as stated in the lemma below.
	\begin{lemma}
		\label{lem:hard}
		If $\msf{In}_{\Phi,X}$ is a positive instance of the $\LTL_\msf{Learn}$ decision problem, then $\Phi$ is satisfiable.
	\end{lemma}
	
	The difficulty to the prove this lemma lies in the fact that from an arbitrary $\LTL$-formula, albeit of size at most $B$, it is not straightforward, a priori, to extract a valuation from it. Let us first handle the simple case where there is 
	a separating $\LTL$-formula whose shape is similar to the $\LTL$-formula that we have exhibited in the proof of Lemma~\ref{lem:easy}.
	\begin{lemma}
		\label{lem:proof_simple_case}
		Assume that there is an $\LTL$-formula $\varphi$ separating $\mathcal{P}$ and $\mathcal{N}$ of the following shape: 
		\begin{equation*}
			\varphi = z_1 \lor z_2 \lor \cdots \lor z_m
		\end{equation*}
		where, for all $k \in \Idm$, we have $z_k \in \{ x_k,\bar{x_k} \}$. Then, there is a valuation satisfying the formula $\Phi$.
	\end{lemma}
	\begin{proof}
		We define the valuation $v_\varphi: X \rightarrow \{ \top,\bot \}$ as follows, for all $k \in \Idm$:
		\begin{align*}
			v_\varphi(x_k) := \begin{cases}
				\top & \text{ if } z_k = x_k \\
				\bot & \text{ if } z_k = \bar{x_k} \\
			\end{cases}
		\end{align*}
		Let $i \in \llbracket 1,n \rrbracket$. We have $c_i \models \varphi$. Hence, there is $1 \leq j \leq 3$ and $1 \leq k \leq m$ such that $\tilde{l^i_j} = z_k$. That is, we have either $l^i_j = x_k = z_k$, it which case $v_\varphi(x_k) = \top$ and $v_\varphi \models l^i_1 \lor l^i_2 \lor l^i_3$, or $l^i_j = \lnot x_k$ and $\bar{x_k} = z_k$, it which case $v_\varphi(x_k) = \bot$ and $v_\varphi \models l^i_1 \lor l^i_2 \lor l^i_3$. Thus, in any case, we have $v_\varphi \models l^i_1 \lor l^i_2 \lor l^i_3$. Since this holds for all $1 \leq i \leq n$, it follows that $v_\varphi \models \Phi$.	
	\end{proof}
	
	Our goal is to show that whenever $\msf{In}_{\Phi,X}$ is a positive instance of the $\LTL_\msf{Learn}$ decision problem, then the premise of Lemma~\ref{lem:proof_simple_case} holds. The method that we take to achieve this takes two stages:
	\begin{itemize}
		\item First, we reduce the set of $\LTL$-formulae that we need to consider with the use of  syntactical requirements. Informally, this consists in proving that if there is an arbitrary $\LTL$-formula of size at most $B$ that accepts all words in $\mathcal{P}$ and rejects all words in $\mathcal{N}$, then there is also one in that set of $\LTL$-formula. 
		\item Second, once restricted to that set of $\LTL$-formulae, we establish, by induction, that these formulae imply (in a specific sense) formulae of the shape described in Lemma~\ref{lem:proof_simple_case}.
	\end{itemize}
	
	\subsubsection{Removing all temporal operators}	
	The first step that we take is to show that we can disregard any temporal operators, i.e. we show that we can restrict ourselves to temporal-free $\LTL$-formulae. This notion is defined below.
	\begin{definition}
		An $\LTL$-formula $\varphi$ is \emph{temporal-free} is none of the operators $\lX,\lF,\lG,\lU,\lW,\lR,\lM$ occur in $\varphi$. 
	\end{definition}
	
	To prove this result, we define a function $\msf{tr}: \LTL \rightarrow \LTL$ on $\LTL$-formulae that removes all temporal operators, that does not increase the size of formulae while doing so and that preserves an equivalence on size-1 ultimately periodic words. This is formally stated below.
	\begin{lemma}
		\label{lem:no_temp}
		There exists a function $\msf{tr}: \LTL \rightarrow \LTL$ such that, for all $\LTL$-formulae $\varphi \in \LTL$:
		\begin{enumerate}
			\item $\msf{SubF}(\msf{tr}(\varphi)) \subseteq \msf{tr}[\msf{SubF}(\varphi)]$, hence $\msf{Sz}(\msf{tr}(\varphi)) \leq \msf{Sz}(\varphi)$;
			\item $\msf{tr}(\varphi)$ is temporal-free;
			\item $\varphi \Leftrightarrow_1 \msf{tr}(\varphi)$, i.e. for all $\alpha \in 2^\prop$, we have $\alpha^\omega \models \varphi \Leftrightarrow \alpha^\omega \models \msf{tr}(\varphi)$.
		\end{enumerate}
	\end{lemma}

	To make the proof of this lemma not too cumbersome, let first state and prove another that guides us on how to translate temporal operators into non-temporal operators, while preserving an equivalence on size-1 ultimately periodic words.
	\begin{lemma}
		\label{lem:equiv_temp}
		Consider two arbitrary $\LTL$-formulae $\varphi,\varphi_1,\varphi_2 \in \LTL$. We have:
		\begin{enumerate}
			\item for all $\bullet \in \{ \lX,\lF,\lG \}$, we have $\varphi \Leftrightarrow_1 \bullet \varphi$;
			\item for all $\bullet \in \{ \lU,\lR \}$, we have $\varphi_2 \Leftrightarrow_1 \varphi_1 \lU \varphi_2$;
			\item we have $\varphi_1 \lor \varphi_2 \Leftrightarrow_1  \varphi_1 \lW \varphi_2$;
			\item we have $\varphi_1 \wedge \varphi_2 \Leftrightarrow_1  \varphi_1 \lM \varphi_2$;
		\end{enumerate}
	\end{lemma} 
	\begin{proof}
		Consider any $\alpha \subseteq \prop$. Note that for all $i \in \N$, we have $\alpha^\omega[i:] = \alpha^\omega$. Therefore:
		\begin{enumerate}
			\item \begin{itemize}
				\item $\alpha^\omega \models \varphi \Leftrightarrow \alpha^\omega[1:] \models \varphi \Leftrightarrow \alpha^\omega \models \lX \varphi$;
				\item $\alpha^\omega \models \varphi \Leftrightarrow \exists i\in \N,\; \alpha^\omega[i:] \models \varphi \Leftrightarrow \alpha^\omega \models \lF \varphi$;
				\item $\alpha^\omega \models \varphi \Leftrightarrow \forall i\in \N,\; \alpha^\omega[i:] \models \varphi \Leftrightarrow \alpha^\omega \models \lG \varphi;$ 
			\end{itemize}
			\item If $\alpha^\omega \models \varphi_2$, then $\alpha^\omega \models \varphi_1 \lU \varphi_2$.  Furthermore, if $\alpha^\omega \models \varphi_1 \lU \varphi_2$, then there is some $i \in \N$ such that $\alpha^\omega[i:] \models \varphi_2$. Hence, $\alpha^\omega \models \varphi_2$. In fact,  $\alpha^\omega \models \varphi_1 \lU \varphi_2 \Leftrightarrow \alpha^\omega \models \varphi_2$.
			
			In addition, we have:
			\begin{equation*}
				\alpha^\omega \models \varphi_1 \lR \varphi_2 \Leftrightarrow \alpha^\omega \not\models \neg\varphi_1 \lU \neg\varphi_2 \Leftrightarrow \alpha^\omega \not\models \neg\varphi_2 \Leftrightarrow \alpha^\omega \models \varphi_2
			\end{equation*}
			\item $\alpha^\omega \models \varphi_1 \lW \varphi_2 \Leftrightarrow \alpha^\omega \models (\varphi_1 \lU \varphi_2) \lor \lG \varphi_1 \Leftrightarrow \varphi_2 \lor \varphi_1$;
			\item $\alpha^\omega \models \varphi_1 \lM \varphi_2 \Leftrightarrow \alpha^\omega \models (\varphi_1 \lR \varphi_2) \wedge \lF \varphi_1 \Leftrightarrow \varphi_2 \wedge \varphi_2$.
		\end{enumerate}
	\end{proof}

	The proof of Lemma~\ref{lem:no_temp} is now straightforward, although somewhat long since there are several cases to tackle.
	\begin{proof}
		We define $\msf{tr}: \LTL \rightarrow \LTL$ by induction on $\LTL$-formulae. Consider a formula $\varphi \in \LTL$:
		\begin{itemize}
			\item Assume that $\varphi = x$ for any $x \in \prop$. We set $\msf{tr}(\varphi) := \varphi$ which satisfies conditions 1.-3.; . 
			\item for all $\bullet \in \msf{Op}_{Un}$, assume that $\varphi = \bullet \varphi'$ and that $\msf{tr}$ is defined on all sub-formulae of $\varphi'$ and satisfies conditions 1.-3. on $\varphi'$. 
			\begin{itemize}
				\item if $\bullet = \lnot$, we set $\msf{tr}(\varphi) := \lnot \msf{tr}(\varphi')$, which ensures that $\msf{SubF}(\msf{tr}(\varphi)) = \{ \lnot \msf{tr}(\varphi') \} \cup \msf{SubF}(\msf{tr}(\varphi')) \subseteq \msf{tr}[\{ \lnot \varphi' \}] \cup \msf{tr}[\msf{SubF}(\varphi')] = \msf{tr}[\msf{SubF}(\varphi)]$. Furthermore, for all $\alpha \subseteq \prop$, we have $\alpha^\omega \models \msf{tr}(\varphi) \Leftrightarrow \alpha^\omega \not\models \msf{tr}(\varphi') \Leftrightarrow \alpha^\omega \not\models \varphi' \Leftrightarrow \alpha^\omega \models \varphi$. Thus, $\varphi \Leftrightarrow_1\msf{tr}(\varphi)$;
				\item if $\bullet \in \{ \lX,\lF,\lG \}$, we set $\msf{tr}(\varphi) := \msf{tr}(\varphi')$, which ensures that $\msf{SubF}(\msf{tr}(\varphi)) = \msf{SubF}(\msf{tr}(\varphi')) \subseteq \msf{tr}[\msf{SubF}(\varphi')] \subseteq \msf{tr}[\msf{SubF}(\varphi)]$. Furthermore, $\varphi \Leftrightarrow_1 \msf{tr}(\varphi)$ by Lemma~\ref{lem:equiv_temp} and since $\varphi' \Leftrightarrow_1 \msf{tr}(\varphi')$.
			\end{itemize}
			\item for any $\bullet \in \msf{Op}_{Bin}$, assume that $\varphi = \varphi_1 \bullet \varphi_2$ and that $\msf{tr}$ is defined on all sub-formulae of $\varphi_1$ and $\varphi_2$ and satisfies conditions 1.-3. on $\varphi_1$ and $\varphi_2$. 
			\begin{itemize}
				\item if $\bullet \in \{ \lor,\wedge,\Rightarrow,\Leftrightarrow \}$, then we set $\msf{tr}(\varphi) := \msf{tr}(\varphi_1) \bullet\msf{tr}(\varphi_2)$, which ensures that $\msf{SubF}(\msf{tr}(\varphi)) = \{ \msf{tr}(\varphi) \} \cup \msf{SubF}(\msf{tr}(\varphi_1)) \cup \msf{SubF}(\msf{tr}(\varphi_2)) \subseteq \msf{tr}[\{ \varphi \}] \cup \msf{tr}[\msf{SubF}(\varphi_1)] \cup \msf{tr}[\msf{SubF}(\varphi_2)] \subseteq \msf{tr}[\msf{SubF}(\varphi)]$. Furthermore, since $\varphi_1 \Leftrightarrow_1 \msf{tr}(\varphi_1)$ and $\varphi_2 \Leftrightarrow_1 \msf{tr}(\varphi_2)$, we have $\varphi \Leftrightarrow_1 \msf{tr}(\varphi)$;
				\item if $\bullet \in \{ \lU,\lR \}$, then we set $\msf{tr}(\varphi) := \msf{tr}(\varphi_2)$, which ensures that $\msf{SubF}(\msf{tr}(\varphi)) = \msf{SubF}(\msf{tr}(\varphi_2)) \subseteq \msf{tr}[\msf{SubF}(\varphi_2)] \subseteq \msf{tr}[\msf{SubF}(\varphi)]$. Furthermore, $\varphi \Leftrightarrow_1 \msf{tr}(\varphi)$ by Lemma~\ref{lem:equiv_temp} and since $\varphi_1 \Leftrightarrow_1 \msf{tr}(\varphi_1)$ and $\varphi_2 \Leftrightarrow_1 \msf{tr}(\varphi_2)$;
				\item if $\bullet \in \{ \lW \}$, then we set $\msf{tr}(\varphi) := \msf{tr}(\varphi_1) \lor \msf{tr}(\varphi_2)$, which ensures that $\msf{SubF}(\msf{tr}(\varphi)) = \{ \msf{tr}(\varphi) \} \cup \msf{SubF}(\msf{tr}(\varphi_1)) \cup \msf{SubF}(\msf{tr}(\varphi_2)) \subseteq \msf{tr}[\{ \varphi \}] \cup \msf{tr}[\msf{SubF}(\varphi_1)] \cup \msf{tr}[\msf{SubF}(\varphi_2)] \subseteq \msf{tr}[\msf{SubF}(\varphi)]$, and $\varphi \Leftrightarrow_1 \msf{tr}(\varphi)$ by Lemma~\ref{lem:equiv_temp} and since $\varphi_1 \Leftrightarrow_1 \msf{tr}(\varphi_1)$ and $\varphi_2 \Leftrightarrow_1 \msf{tr}(\varphi_2)$;
				\item if $\bullet \in \{ \lM \}$, then we set $\msf{tr}(\varphi) := \msf{tr}(\varphi_1) \wedge \msf{tr}(\varphi_2)$, which ensures that $\msf{SubF}(\msf{tr}(\varphi)) = \{ \msf{tr}(\varphi) \} \cup \msf{SubF}(\msf{tr}(\varphi_1)) \cup \msf{SubF}(\msf{tr}(\varphi_2)) \subseteq \msf{tr}[\{ \varphi \}] \cup \msf{tr}[\msf{SubF}(\varphi_1)] \cup \msf{tr}[\msf{SubF}(\varphi_2)] \subseteq \msf{tr}[\msf{SubF}(\varphi)]$ and $\varphi \Leftrightarrow_1 \msf{tr}(\varphi)$ by Lemma~\ref{lem:equiv_temp} and since $\varphi_1 \Leftrightarrow_1 \msf{tr}(\varphi_1)$ and $\varphi_2 \Leftrightarrow_1 \msf{tr}(\varphi_2)$;
				%
				%
			\end{itemize}
		\end{itemize}
		This concludes the proof since, by construction, for all $\LTL$-formulae $\varphi$, no temporal operator occur in the formula $\msf{tr}(\varphi)$.
	\end{proof}
	
	\subsubsection{$\LTL$-formulae with at least $k$ different propositions}
	Second, we show that for a formula $\varphi$ to accept all words in $\mathcal{P}$ and reject all words in $\mathcal{N}$, then it must be that that, for all $k \in \Idm$, at least one variable among $\{ x_k,\bar{x_k} \}$ must occur in $\varphi$. Let first us define below a few notations that will be useful in the following.
	\begin{definition}
		\label{def:relevant_def}
		For all $k \in \Idm$, we let $X_k := \{ x_k,\bar{x_k} \} \subseteq \prop$. Furthermore, for all $\I \subseteq \Idm$, we let $X_\I \subseteq \prop$ denote the set of propositions $X_\I := \cup_{k \in \I} X_k$.
		
		Consider an $\LTL$-formula $\varphi$. For all subsets of propositions $Y \subseteq \prop$, we say that $Y$ occurs in $\varphi$ if a proposition in $Y$ is a sub-formula of $\varphi$, i.e. $Y \cap \prop(\varphi) \neq \emptyset$.
		
		Furthermore, for all $S \subseteq (2^{\prop})^\omega$, we say that the formula $\varphi$ accepts (resp. rejects) $S$ if $\varphi$ accepts all words in $S$ (resp. rejects all words in $S$). We say that $\varphi$ distinguishes two sets of words if it accepts one set and rejects the other.
	\end{definition}
	
	We can now state the lemma of interest for us. 
	\begin{lemma}
		\label{lem:distinguish}
		Consider a temporal-free $\LTL$-formula $\varphi$. Consider three subsets of propositions $\alpha_1,\alpha_2,Y \subseteq \prop$, and assume that, for all $x \in \prop \setminus Y$, we have $x \in \alpha_1$ if and only if $x \in \alpha_2$. In that case, if the formula $\varphi$ distinguishes the words $\alpha_1^\omega$ and $\alpha_2^\omega$, then the set $Y$ occurs in $\varphi$.
	\end{lemma}
	\begin{proof}
		Let us prove this property $\mathcal{P}(\varphi)$ on $\LTL$-formulae $\varphi$ by induction. Consider such an $\LTL$-formula $\varphi$:
		\begin{itemize}
			\item Assume that $\varphi = x$ for some $x \in \prop$. Then, if $\varphi$ distinguishes $\alpha_1^\omega$ and $\alpha_2^\omega$, it must be that $x \in \prop$, by assumption.
			\item Assume that $\varphi = \neg \varphi'$ and $\mathcal{P}(\varphi')$ holds. Then, if $\varphi$ distinguishes $\alpha_1^\omega$ and $\alpha_2^\omega$, so does $\varphi'$. Hence, $\mathcal{P}(\varphi)$ follows.
			\item Finally, assume that $\varphi = \varphi_1 \bullet \varphi_2$ for some $\bullet \in \{ \lor,\wedge,\Rightarrow,\Leftrightarrow \}$ and that $\mathcal{P}(\varphi_1)$ and $\mathcal{P}(\varphi_2)$ hold. Then, if neither $\varphi_1$ nor $\varphi_2$ distinguish $\alpha_1^\omega$ and $\alpha_2^\omega$, neither does $\varphi$. Hence, assuming that $\varphi$ distinguish $\alpha_1^\omega$ and $\alpha_2^\omega$, then $\prop \setminus Y$ occurs in $\varphi_1$ or $\varphi_2$, and it therefore also occurs in $\varphi$. Hence, $\mathcal{P}(\varphi)$ holds.
		\end{itemize}
	\end{proof}

	When applied to the reduction that we have defined in Definition~\ref{def:reduction}, we obtain the corollary below.
	\begin{corollary}
		\label{coro:at_least_k}
		Assume that a temporal-free $\LTL$-formula distinguishes $\mathcal{P}$ and $\mathcal{N}$. Then, for all $k \in \Idm$, a variable in $X_k$ must occur in $\varphi$.
	\end{corollary}
	\begin{proof}
		Let $k \in \Idm$. We have $\rho_k = X_k^\omega$ and $\eta = \emptyset^\omega$. Therefore, by Lemma~\ref{lem:distinguish}, if an $\LTL$-formula $\varphi$ distinguishes $\rho_k$ and $\eta$, it must be that $X_k$ occurs in $\varphi$. Since this holds for all $k \in \Idm$ and for all $k,k' \in \Idm$, we have $X_k \cap X_{k'} \neq \emptyset \Rightarrow k = k'$, the corollary follows.
	\end{proof}
	
	\subsubsection{Concise formulae}
	With the lemmas we have proven so far, we know that we can restrict ourselves to $\LTL$-formulae:
	\begin{itemize}
		\item that are temporal-free;
		\item where, for all $k \in \Idm$, a variable in $X_k$ occurs.
	\end{itemize}
	In addition, recall that the formulae that we consider are of size at most $B = 2m-1$. Our next goal is to establish that, for all $k \in \Idm$, when at least $k$ different propositions occur in a formula of size at most $2k-1$, then that formula must  be \emph{concise}. Informally, an $\LTL$-formula is concise if it never uses any unary operators and never uses several times a sub-formula. We define formally (by induction) this notion below. 
	
	\begin{definition}
		\label{def:concise_formulae}
		For all $\LTL$-formulae, we let $\prop(\varphi) = \msf{SubF}(\varphi) \cap \prop$ denote the of propositions occurring in $\varphi$. Then, we define by induction the \emph{concise} $\LTL$-formulae:
		\begin{itemize}
			\item for all $x \in \prop$, the formula $\varphi := x$ is \emph{concise};
			\item given any two concise formulae $\varphi_1$ and $\varphi_2$ such that $\prop(\varphi_1) \cap \prop(\varphi_2) = \emptyset$, then the formula $\varphi := \varphi_1 \bullet \varphi_2$ is concise for any $\bullet \in \msf{Op}_{Bin}$.
		\end{itemize}
		We denote by $\LTL_\msf{cs}$ the set of concise $\LTL$-formulae. 
		
		For all $Y \subseteq \prop$, we say that an $\LTL$-formula $\varphi$  \emph{represents concisely} $Y$ if $Y \subseteq \prop(\varphi)$ and $\msf{Sz}(\varphi) \leq 2|Y| - 1$. 
	\end{definition}

	Our goal is now to show the lemma below. 
	\begin{lemma}
		\label{lem:concise}
		For all $Y \subseteq \prop$, if an $\LTL$-formula $\varphi$ represents concisely $Y$ then it is concise and $Y = \prop(\varphi)$.
	\end{lemma}

	In order to prove this lemma, we are first going to prove the lemma below, relating in a given $\LTL$-formula, the number of sub-formulae and the number of different propositions that occur.
	\begin{lemma}
		\label{lem:at_least_that_many_subformulae}
		For all $Y \subseteq \prop$, we let $\msf{Occ}(Y),\msf{Occ}^+(Y)
		,\mathsf{NbSubF}(Y),\mathsf{NbSubF}^+(Y)		$ denote the properties on $\LTL$-formulae such that an $\LTL$-formula $\varphi$ satisfies:
		\begin{itemize}
			\item $\msf{Occ}(Y)$ if all variables in $Y$ occur in $\varphi$, i.e. $Y \subseteq \prop(\varphi)$;
			\item $\msf{Occ}^+(Y)$ if all variables in $Y$ occur in $\varphi$, i.e. $Y \subsetneq \prop(\varphi)$;
			\item $\mathsf{NbSubF}(Y)$ if there are at least $2|Y|-1$ different sub-formulae of $\varphi$ where $Y$ occurs;
			\item $\mathsf{NbSubF}^+(Y)$ if there are at least $2|Y|$ different sub-formulae of $\varphi$ where $Y$ occurs.
		\end{itemize}

		Then, consider any $\LTL$-formula $\varphi$ on $\prop$. If $\varphi$ satisfies $\msf{Occ}(Y)$ (resp. $\msf{Occ}^+(Y)$), then it also satisfies $\mathsf{NbSubF}(Y)$ (resp. $\mathsf{NbSubF}^+(Y)$).
	\end{lemma}
	\begin{proof}
		Let us show this lemma by induction on $\LTL$-formulae $\varphi$:
		\begin{itemize}
			\item Assume that $\varphi = x$. Then, $\varphi$ satisfies $\msf{Occ}(Y)$ only for $Y = \emptyset$ and $Y = \{ x \}$, and it also satisfies $\mathsf{NbSubF}(\emptyset)$ and $\mathsf{NbSubF}(\{x\})$. Furthermore, $\varphi$ satisfies $\msf{Occ}^+(Y)$ only for $Y = \emptyset$, and it also satisfies $\mathsf{NbSubF}^+(\emptyset)$;
			\item for all $\bullet \in \msf{Op}_{Un}$, assume that $\varphi = \bullet \varphi'$. Consider any $Y \subseteq \prop$. If $\varphi$ satisfies $\msf{Occ}(Y)$ (resp. $\msf{Occ}^+(Y)$), then so does $\varphi'$. Hence, $\varphi'$ satisfies $\mathsf{NbSubF}(Y)$ (resp. $\msf{NbSubF}^+(Y)$), and therefore so does $\varphi$;
			\item for all $\bullet \in \msf{Op}_{Bin}$, assume that $\varphi = \varphi_1 \bullet \varphi_2$. Consider any $Y \subseteq \prop$ and assume that $\varphi$ satisfies $\msf{Occ}(Y)$. Let $Y_1 := \prop(\varphi_1) \cap Y \subseteq Y$ denote the set of different propositions in $Y$ occurring in $\varphi_1$ (they may also occur in $\varphi_2$). Let also $Y_2' := (\prop(\varphi_2) \setminus \prop(\varphi_1)) \cap Y \subseteq Y$ denote the set of different propositions in $Y$  occurring in $\varphi_2$ and not in $\varphi_1$. We have $Y = Y_1 \cup Y_2'$. Furthermore:
			\begin{itemize}
				\item by our induction hypothesis on $\varphi_1$, there are at least $k_1 := 2|Y_1|-1$ sub-formulae in $\msf{SubF}(\varphi_1) \subseteq \msf{SubF}(\varphi)$ where $Y_1$ occurs;
				\item by our induction hypothesis on $\varphi_2$, there are also at least $k_2' := 2|Y_2'|-1$ sub-formulae in $\msf{SubF}(\varphi_2) \subseteq \msf{SubF}(\varphi)$ where $Y'_2$ occurs. By definition of $Y'_2$, it follows that all these sub-formulae are not sub-formulae of $\varphi_1$;
				\item Finally, the sub-formula $\varphi$ itself is a sub-formula of $\varphi$ that is neither a sub-formula of $\varphi_1$ nor of $\varphi_2$ and where $Y$ occurs.
			\end{itemize}
			Therefore, there are at least $k_1 + k_2' + 1 = 2|Y_1|-1 + 2|Y'_2|-1 + 1 \geq 2|Y|-1$ different sub-formulae of $\varphi$ where $Y$ occurs. That is, $\varphi$ satisfies $\msf{NbSubF}(Y)$. Furthermore, if $\varphi$ satisfies also $\msf{Occ}^+(Y)$, then $\prop \setminus Y$ occurs in $\varphi_1$ or $\varphi_2$. Therefore, the first item holds with $k_1 = 2|Y_1|$ or the second item holds with $k_2' = 2|Y_2'|$. In both cases, we obtain that they are at least $k_1 + k_2' + 1 \geq 2|Y|$ different sub-formulae of $\varphi$ where $Y$ occurs. That is, $\varphi$ satisfies $\msf{NbSubF}^+(Y)$.
		\end{itemize}
	\end{proof}
	We can now proceed to the proof of Lemma~\ref{lem:concise}.
	\begin{proof}
		We show by induction on $\LTL$-formulae $\varphi$ the following property $\mathcal{P}(\varphi)$: for all $Y \subseteq \prop$, if $\varphi$ represents concisely $Y$, then it is concise. Consider an $\LTL$-formula $\varphi$:
		\begin{itemize}
			\item if $\varphi = x$ then $\mathcal{P}(\varphi)$ straightforwardly holds;
			\item for all $\bullet \in \msf{Op}_{Un}$, assume that $\varphi = \bullet \varphi'$. Consider any $Y \subseteq \prop$. If $\varphi$ represents concisely $Y$, then $Y \subseteq \prop(\varphi')$. By Lemma~\ref{lem:at_least_that_many_subformulae}, it must be that $\msf{Sz}(\varphi') \geq 2|Y|-1$, hence $\msf{Sz}(\varphi) \geq 2|Y|$, which is a contradiction since $\varphi$ represents concisely $Y$, and hence $\msf{Sz}(\varphi) \leq 2|Y|-1$. In fact, in that case, $\varphi$ does not represent concisely $Y$ for any $Y \subseteq \prop$;
			\item for all $\bullet \in \msf{Op}_{Bin}$, assume that $\varphi = \varphi_1 \bullet \varphi_2$. Consider any $Y \subseteq \prop$ and assume that $\varphi$ represents concisely $Y$. Let $k := |Y|$. Let $Y_{1} := \prop(\varphi_1) \setminus \prop(\varphi_2)$, $Y_{2} := \prop(\varphi_2) \setminus \prop(\varphi_1)$ and $Y_{1,2} := \prop(\varphi_1) \cap \prop(\varphi_2)$. Let $k_1 := |Y_1|$, $k_2 := |Y_2|$ and $k_{1,2} := |Y_{1,2}|$. We have $k_1 + k_2 + k_{1,2} \geq k$ and $|\msf{SubF}(\varphi)| = 1 + |\msf{SubF}(\varphi_1) \cup \msf{SubF}(\varphi_2)|$, hence $|\msf{SubF}(\varphi_1) \cup \msf{SubF}(\varphi_2)| \leq 2k-2$. 
			
			Now, let us assume towards a contradiction that $Y_{1,2} \neq \emptyset$. In that case, by Lemma~\ref{lem:at_least_that_many_subformulae}, since $\varphi_1$ satisfies $\msf{Occ}^+(Y_1)$, it follows that there are at least $2 k_1$ sub-formulae of $\varphi_1$ where $X_1$ occurs. Recall that $X_1$ does not occur in $\varphi_2$. It follows that:
				\begin{equation*}
				|\msf{SubF}(\varphi_2)| \leq 2k-2 - 2 k_1 \leq 2(k_1 + k_2 + k_{1,2}) - 2 - 2 k_1 = 2(k_2 + k_{1,2}) - 2
				\end{equation*}
			However, by Lemma~\ref{lem:at_least_that_many_subformulae}, since $\varphi_2$ satisfies $\msf{Occ}(Y_{1,2} \cup Y_2)$, it follows that $\varphi_2$ has at least $2 (k_2 + k_{1,2}) -1$ sub-formulae, which is a contradiction. In fact, $Y_{1,2} = \emptyset$. 
			
			By Lemma~\ref{lem:at_least_that_many_subformulae}, there are at least $2k_1 -1$ (resp. $2k_2 -1$) sub-formulae in $\varphi_1$ (resp. $\varphi_2$) where $Y_1$ (resp. $Y_2$) occurs. Therefore: 
			\begin{equation*}
				2k-1 \geq |\msf{SubF}(\varphi)| = 1 + |\msf{SubF}(\varphi_1) \cup \msf{SubF}(\varphi_2)| \geq 1 + 2k_1 - 1 + 2k_2-1 = 2(k_1 + k_2)-1
			\end{equation*}
			We obtain that $k = k_1 + k_2$, $|\msf{SubF}(\varphi)| = 2k-1$, $|\msf{SubF}(\varphi_1)| = 2k_1 -1$ and $|\msf{SubF}(\varphi_2)| = 2k_2 -1$. Therefore: 
			\begin{itemize}
				\item $\msf{Sz}(\varphi_1) =  |\msf{SubF}(\varphi_1)| = 2k_1 -1$, hence $\varphi_1$ represents concisely $Y_1$. Hence, by our induction hypothesis, $\varphi_1$ is concise;
				\item similarly, $\msf{Sz}(\varphi_2) =  |\msf{SubF}(\varphi_2)| = 2k_2 -1$, hence $\varphi_2$ represents concisely $Y_2$. Hence, by our induction hypothesis, $\varphi_2$ is concise;
				\item 
				$Y_{1,2} = \prop(\varphi_1) \cap \prop(\varphi_2) = \emptyset$.
			\end{itemize}
			Hence, $\varphi$ is concise. 
		\end{itemize}
		We can conclude that $\mathcal{P}(\varphi)$ holds for all $\LTL$-formulae $\varphi$.
		
		Consider now any $\LTL$-formula $\varphi$ and any $Y \subseteq \prop$. If $\varphi$ represents concisely $Y$, then by $\mathcal{P}(\varphi)$, we have that $\varphi$ is concise. Furthermore, by Lemma~\ref{lem:at_least_that_many_subformulae}, it must be that $\prop(\varphi) \setminus Y = \emptyset$ since $\msf{Sz}(\varphi) \leq 2|Y|-1$. Therefore, $Y = \prop(\varphi)$ and $\msf{Sz}(\varphi) = 2|Y|-1$.
	\end{proof}
	
	\subsubsection{Disjunctive $\LTL$-formulae}
	The lemmas that we have proven so far allow us to restrict the kind of $\LTL$-formulae that we need to consider. We know that it is sufficient to consider the $\LTL$-formulae satisfying the following: 
	\begin{itemize}
		\item they are are temporal-free;
		\item for all $k \in \Idm$, a variable in $X_k$ occurs;
		\item they are concise.
	\end{itemize}
	However, these restrictions are not enough to prove that the premise of Lemma~\ref{lem:proof_simple_case} holds, since the $\LTL$-formulae that we need to consider may use any of the four binary logical connectors. 
	
	Let us give below the definition of the $\LTL$-formulae that are interested in.
	\begin{definition}
		\label{def:disjunctive}
		Consider some $\I \subseteq \Idm$. An $\LTL$-formula $\varphi$ is \emph{disjunctive on $\I$} if:
		\begin{itemize}
			\item $\varphi$ is temporal free;
			\item $\varphi$ is concise;
			\item the only binary logical operator occurring in $\varphi$ is the $\lor$ operator;
			\item $|\prop(\varphi)| = |\I|$ and for all $k \in \I$, we have $|\prop(\varphi) \cap X_k| = 1$.
		\end{itemize}
	\end{definition}
	
	A simple proof by induction shows that if an $\LTL$-formula $\varphi$ is disjunctive on $\Idm$, then it is of the shape: 
	\begin{equation*}
		\varphi = z_1 \lor z_2 \lor \cdots \lor z_m
	\end{equation*}
	where, for all $k \in \Idm$, we have $z_k \in X_k = \{ x_k,\bar{x_k} \}$.
	
	Then, we define what it means for us that an $\LTL$-formula implies a disjunction.
	\begin{definition}
		Consider any $\I \subseteq \Idm$ and an $\LTL$-formula $\varphi$. We say that $\varphi$ \emph{implies a disjunction on $\I$} if there exists an $\LTL$-formula $\varphi'$ that is a disjunction on $\I$ such that, for all $\alpha \subseteq X_\I$, if $\alpha^\omega \models \varphi$, then $\alpha \models \varphi'$.
		
		Furthermore, we say that $\varphi$ is \emph{$\I$-concise} if it is concise, $|\prop(\varphi)| = |\I|$ and for all $k \in \I$, we have $|\prop(\varphi) \cap X_k| = 1$.
	\end{definition}
	
	Our goal is now to establish the lemma below.
	\begin{lemma}
		\label{lem:implies_disjunction}
		Consider a temporal-free $\LTL$-formula $\varphi$ that is $\Idm$-concise and assume that it accepts $\mathcal{P}$ and rejects $\mathcal{N}$. Then, $\varphi$ implies a disjunction on $\Idm$. 
	\end{lemma}
	It may seem that the proof of this lemma is a straightforward application of well-known results of propositional logic about formulae in conjunctive normal form (CNF). Indeed, since we are considering a temporal-free $\LTL$-formula, it can be seen as a formula in propositional logic. In this logic, the fact that a formula implies a disjunction is well known: given the expression of this formula in CNF, any clause is a disjunction implied by the formula. However, there are crucial differences here: first, in the disjunctions that we consider, there cannot be any negation. Nonetheless, one can say that $\lnot x_1$ is equivalent to $\bar{x_1}$ and $\lnot \bar{x_1}$ is equivalent to $x_1$, since the proposition $\bar{x}$ is meant as a stand in for $\lnot x$. But the above equivalences do not hold with the $\LTL$ semantics. The reason for that is that $\LTL$-formulae are evaluated on subsets of propositions that may not consider at all a variable in $X$ (i.e. neither $x_k$ not $\bar{x_k}$ appears in the subset for some $k \in \Idm$). Hence, it is not straightforward, a priori, to come up with a disjunction that is implied by $\LTL$-formulae of interest. Let us illustrate this discussion on an example.
	\begin{example}
		Let us come back to the $\msf{SAT}$-formula from Example~\ref{ex:1}. Consider the $\LTL$-formula $\varphi$ below:
		\begin{equation*}
			\varphi := (x_1 \Leftrightarrow x_2) \Leftrightarrow x_3
		\end{equation*}
		One can check that this formula accepts $\mathcal{P}$ and rejects $\mathcal{N}$, as defined in Example~\ref{ex:1}. Furthermore, in propositional logic, this formula implies the disjunction defined below:
		\begin{equation*}
			C = \lnot x_1 \lor \lnot x_2 \lor x_3 
		\end{equation*}

		Let us first explain why we do not allow negations in the disjunctions that we consider. To do so, let us assume that the $\msf{SAT}$-formula $\Phi'$ that we consider is defined by $\Phi' = \Phi \wedge \Phi_3$ where $\Phi_3$ is either equal to $x_3$ or $\lnot x_3$. Then, the formula $C$ would accept $\mathcal{P}$ and reject $\mathcal{N}$ regardless of what $\Phi_3$ is equal to, which is not what we want. 
		In a way, $C$ accepts $\Phi_3$ not because it has properly chosen either $x_3$ or $\bar{x_3}$ but because it has negated another variable ($x_1$ and $x_2$). 
		
		This is why we do not want negations in the disjunctions that we consider. Hence, let us write $C$ as a disjunction without negation. We obtain: 
		\begin{equation*}
			\tilde{C} =  \bar{x_1} \lor \bar{x_2} \lor x_3
		\end{equation*}
		However, it is now not true that $\varphi$ implies $\tilde{C}$. This is witnessed e.g. by the infinite word $(\{ x_1 \})^\omega$ that is accepted by $\varphi$ and rejected by $\tilde{C}$. 
		
		By the way, one can check that the formula $\varphi$ implies on $\{ 1,2,3\}$ the disjunction $x_1 \lor x_2 \lor x_3$.
	\end{example}

	Nonetheless the proof of Lemma~\ref{lem:implies_disjunction} is not complicated, though it is quite long as it involves a proof by induction, with several cases to handle. 
	\begin{proof}
		Let us show by induction on temporal-free concise formulae $\varphi$ the property $\mathcal{P}(\varphi)$: for all $\I \subseteq \Idm$, if $\varphi$ is $\I$-concise, then:
		\begin{enumerate}
			\item if $\varphi$ accepts $\mathcal{P}_\I$ and rejects $\mathcal{N}$, then $\varphi$ implies a disjunction on $\I$;
			\item if $\varphi$ accepts $\mathcal{N}$ and rejects $\mathcal{P}_\I$, then $\neg\varphi$ implies a disjunction on $\I$.
		\end{enumerate}
	
		Straightforwardly, for any concise formula $\varphi = x \in \prop$, $\mathcal{P}(\varphi)$ holds. Indeed, for all $\I \subseteq \Idm$, if $\varphi$ is $\I$-concise, then $\varphi$ itself is a disjunction on $\I$. 
		
		Consider now any concise $\LTL$-formula $\varphi = \varphi_1 \bullet \varphi_2$ for some $\bullet \in \{ \lor,\wedge,\Rightarrow,\Leftrightarrow \}$. Assume that $\varphi$ is $\I$-concise for some $\I \subseteq \Idm$, and that $\mathcal{P}(\varphi_1)$ and $\mathcal{P}(\varphi_2)$ hold. Then, it must be that $\varphi_1$ (resp. $\varphi_2$) is $\I_1$-concise (resp. $\I_2$-concise) for some $\I_1 \subseteq \Idm$ (resp. $\I_2 \subseteq \Idm$) with $\I = \I_1 \cup \I_2$ and $\I_1 \cap \I_2 = \emptyset$. Note that this holds because $\varphi$ is $\I$-concise. Furthermore, this implies that given any disjunctions $\varphi^d_1$ on $\I_1$ and $\varphi^d_2$ on $\I_2$, the $\LTL$-formula $\varphi^d = \varphi^d_1 \lor \varphi^d_2$ is a disjunction on $\I$.

		Note that, by Lemma~\ref{lem:distinguish}, $\varphi_1$ does not distinguish $\rho_k$ and $\rho_{k'}$ for any two $k,k' \in \I_2$. Furthermore, again by Lemma~\ref{lem:distinguish}, $\varphi_1$ does not distinguish $\rho_k$ and $\eta$ for any $k \in \I_2$. Hence, $\varphi_1$ accepts or rejects $\mathcal{P}_{\I_2} \cup \mathcal{N}$. Symmetrically, $\varphi_2$ accepts or rejects $\mathcal{P}_{\I_1} \cup \mathcal{N}$. Finally, for any subset $\I' \subseteq \I$ and $\alpha \in X_\I$, we let $\alpha_{\I'} := \alpha \cap X_{\I'} \subseteq X_{\I'}$. 
		
		Now, they are four cases:
		\begin{itemize}
			\item Assume that $\varphi = \varphi_1 \lor \varphi_2$. 
			\begin{enumerate}
				\item Assume that $\varphi$ accepts $\mathcal{P}_\I$ and rejects $\mathcal{N}$. In that case, we have:
				\begin{itemize}
					\item $\varphi_1$ rejects $\mathcal{N}$, hence it rejects $\mathcal{P}_{\I_2} \subseteq \mathcal{P}_\I$, hence $\varphi_2$ accepts $\mathcal{P}_{\I_2}$;
					\item $\varphi_2$ rejects $\mathcal{N}$, hence it rejects $\mathcal{P}_{\I_1} \subseteq \mathcal{P}_\I$, hence $\varphi_1$ accepts $\mathcal{P}_{\I_1}$.
				\end{itemize}
				Therefore $\varphi_1$ (resp. $\varphi_2$) accepts $\mathcal{P}_{\I_1}$ (resp. $\mathcal{P}_{\I_2}$) and rejects $\mathcal{N}$. By our induction hypothesis, $\varphi_1$ (resp. $\varphi_2$) implies a disjunction $\varphi^d_1$ on $\I_1$ (resp. $\varphi^d_2$ on $\I_2$). Let $\varphi^d := \varphi^d_1 \lor \varphi^d_2$ be a disjunction on $\I$. Let us show that $\varphi$ implies it. 
				
				Consider any $\alpha \subseteq X_{\I}$ and assume that $\alpha^\omega \models \varphi$. Then, $\alpha^\omega \models \varphi_1$ or $\alpha^\omega \models \varphi_2$. Assume that $\alpha^\omega \models \varphi_1$. By Lemma~\ref{lem:distinguish}, it follows that $(\alpha_{\I_1})^\omega \models \varphi_1$ with $\alpha_{\I_1} \subseteq X_{\I_1}$. Hence, $\alpha_{\I_1} \models \varphi^d_1$ and $\alpha_{\I_1} \models \varphi^d$. Therefore, $\alpha \models \varphi^d$ since $\varphi^d$ is a disjunction. This is similar if $\alpha^\omega \models \varphi_2$. It follows that $\varphi$ implies a disjunction on $\I$.
				\item Assume towards a contradiction that $\varphi$ rejects $\mathcal{P}_\I$ and accepts $\mathcal{N}$. In that case:
				\begin{itemize}
					\item $\varphi_1$ rejects $\mathcal{P}_{\I_2} \subseteq \mathcal{P}_{\I}$, hence it rejects $\mathcal{N}$;
					\item $\varphi_2$ rejects $\mathcal{P}_{\I_1} \subseteq \mathcal{P}_{\I}$, hence it rejects $\mathcal{N}$.
				\end{itemize}
				Therefore, $\varphi$ rejects $\mathcal{N}$. Hence, the contradiction. 
			\end{enumerate}
			\item Assume that $\varphi = \varphi_1 \wedge \varphi_2$. 
			\begin{enumerate}
				\item Assume towards a contradiction that $\varphi$ accepts $\mathcal{P}_\I$ and rejects $\mathcal{N}$. In that case:
				\begin{itemize}
					\item $\varphi_1$ accepts $\mathcal{P}_{\I_2} \subseteq \mathcal{P}_{\I}$, hence it accepts $\mathcal{N}$;
					\item $\varphi_2$ accepts $\mathcal{P}_{\I_1} \subseteq \mathcal{P}_{\I}$, hence it accepts $\mathcal{N}$.
				\end{itemize}
				Therefore, $\varphi$ accepts $\mathcal{N}$. Hence, the contradiction.
				\item Assume that $\varphi$ rejects $\mathcal{P}_\I$ and accepts $\mathcal{N}$. In that case, we have:
				\begin{itemize}
					\item $\varphi_1$ accepts $\mathcal{N}$, hence it accepts $\mathcal{P}_{\I_2} \subseteq \mathcal{P}_\I$, hence $\varphi_2$ rejects $\mathcal{P}_{\I_2}$;
					\item $\varphi_2$ accepts $\mathcal{N}$, hence it accepts $\mathcal{P}_{\I_1} \subseteq \mathcal{P}_\I$, hence $\varphi_1$ rejects $\mathcal{P}_{\I_1}$.
				\end{itemize}
				Therefore $\varphi_1$ (resp. $\varphi_2$) rejects $\mathcal{P}_{\I_1}$ (resp. $\mathcal{P}_{\I_2}$) and accepts $\mathcal{N}$. By our induction hypothesis, $\neg \varphi_1$ (resp. $\neg \varphi_2$) implies a disjunction $\varphi^d_{\neg 1}$ on $\I_1$ (resp. $\varphi^d_2$ on $\I_{\neg 2}$). Let $\varphi^d_{\neg} := \varphi^d_{\neg 1} \lor \varphi^d_{\neg 2}$ be a disjunction on $\I$. Let us show that $\neg \varphi$ implies it. 
				
				Consider any $\alpha \subseteq X_{\I}$ and assume that $\alpha^\omega \models \neg \varphi$. Then, $\alpha^\omega \models \neg \varphi_1$ or $\alpha^\omega \models \neg \varphi_2$. Assume that $\alpha^\omega \models \neg \varphi_1$. By Lemma~\ref{lem:distinguish}, it follows that $(\alpha_{\I_1})^\omega \models \neg \varphi_1$. with $\alpha_{\I_1} \subseteq X_{\I_1}$. Hence, $\alpha_{\I_1} \models \varphi^d_{\neg 1}$ and $\alpha_{\I_1} \models \varphi^d_{\neg}$. Therefore $\alpha \models \varphi^d_{\neg}$ since $\varphi^d_{\neg}$ is a disjunction. This is similar if $\alpha^\omega \models \neg\varphi_2$. It follows that $\varphi$ implies a disjunction on $\I$.
			\end{enumerate}
			\item Assume that $\varphi = \varphi_1 \Rightarrow \varphi_2$. 
			\begin{enumerate}
				\item Assume that $\varphi$ accepts $\mathcal{P}_\I$ and rejects $\mathcal{N}$. In that case, we have:
				\begin{itemize}
					\item $\varphi_1$ accepts $\mathcal{N}$, hence it accepts $\mathcal{P}_{\I_2} \subseteq \mathcal{P}_\I$, hence $\varphi_2$ accepts $\mathcal{P}_{\I_2}$;
					\item $\varphi_2$ rejects $\mathcal{N}$, hence it rejects $\mathcal{P}_{\I_1} \subseteq \mathcal{P}_\I$, hence $\varphi_1$ rejects $\mathcal{P}_{\I_1}$.
				\end{itemize}
				Therefore $\varphi_1$ rejects $\mathcal{P}_{\I_1}$ and accepts $\mathcal{N}$. By our induction hypothesis, $\neg \varphi_1$ implies a disjunction $\varphi^d_{\neg 1}$ on $\I_1$. Furthermore, $\varphi_2$ accepts $\mathcal{P}_{\I_2}$ and rejects $\mathcal{N}$. By our induction hypothesis, $\varphi_2$ implies a disjunction $\varphi^d_2$ on $\I_2$. We let $\varphi^d := \varphi^d_{\neg 1} \lor \varphi^d_2$ be a disjunction on $\I$. Let us show that $\varphi$ implies it. 
				
				Consider any $\alpha \subseteq X_{\I}$ and assume that $\alpha^\omega \models \varphi$. Then, $\alpha^\omega \models \neg \varphi_1$ or $\alpha^\omega \models \varphi_2$. Assume that $\alpha^\omega \models \neg \varphi_1$. By Lemma~\ref{lem:distinguish}, it follows that $(\alpha_{\I_1})^\omega \models \neg \varphi_1$. with $\alpha_{\I_1} \subseteq X_{\I_1}$. Hence, $\alpha_{\I_1} \models \varphi^d_{\neg 1}$ and $\alpha_{\I_1} \models \varphi^d$. Therefore $\alpha \models \varphi^d$ since $\varphi^d$ is a disjunction. This is similar if $\alpha^\omega \models \varphi_2$. It follows that $\varphi$ implies a disjunction on $\I$.
				\item Assume towards a contradiction that $\varphi$ rejects $\mathcal{P}_\I$ and accepts $\mathcal{N}$. In that case:
				\begin{itemize}
					\item $\varphi_1$ accepts $\mathcal{P}_{\I_2} \subseteq \mathcal{P}_{\I}$, hence it accepts $\mathcal{N}$;
					\item $\varphi_2$ rejects $\mathcal{P}_{\I_1} \subseteq \mathcal{P}_{\I}$, hence it rejects $\mathcal{N}$.
				\end{itemize}
				Therefore, $\varphi$ rejects $\mathcal{N}$. Hence, the contradiction. 
			\end{enumerate}
			\item Assume that $\varphi = \varphi_1 \Leftrightarrow \varphi_2$. 
			\begin{enumerate}
				\item Assume that $\varphi$ accepts $\mathcal{P}_\I$ and rejects $\mathcal{N}$. In that case, we have:
				\begin{enumerate}
					\item Assume that $\varphi_1$ accepts $\mathcal{N}$. In that case, we have: 
					\begin{itemize}
						\item $\varphi_1$ accepts $\mathcal{N}$, hence it accepts $\mathcal{P}_{\I_2} \subseteq \mathcal{P}_\I$, hence $\varphi_2$ accepts $\mathcal{P}_{\I_2}$;
						\item $\varphi_2$ rejects $\mathcal{N}$, hence it rejects $\mathcal{P}_{\I_1} \subseteq \mathcal{P}_\I$, hence $\varphi_1$ rejects $\mathcal{P}_{\I_1}$.
					\end{itemize}
					Therefore $\varphi_1$ rejects $\mathcal{P}_{\I_1}$ and accepts $\mathcal{N}$. By our induction hypothesis, $\neg \varphi_1$ implies a disjunction $\varphi^d_{\neg 1}$ on $\I_1$. Furthermore, $\varphi_2$ accepts $\mathcal{P}_{\I_2}$ and rejects $\mathcal{N}$. By our induction hypothesis, $\varphi_2$ defines a valuation $\varphi^d_2$ on $\I_2$. We let $\varphi^d := \varphi^d_{\neg 1} \lor \varphi^d_2$ be a disjunction on $\I$. Let us show that $\varphi$ implies it. 
					
					Consider any $\alpha \subseteq X_{\I}$ and assume that $\alpha^\omega \models \varphi$. If $\alpha^\omega \models \varphi_2$, then by Lemma~\ref{lem:distinguish}, it follows that $(\alpha_{\I_2})^\omega \models \varphi_2$. with $\alpha_{\I_2} \subseteq X_{\I_2}$. Hence, $\alpha_{\I_2} \models \varphi^d_2$ and $\alpha_{\I_2} \models \varphi^d$. Therefore $\alpha \models \varphi^d$ since $\varphi^d$ is a disjunction.  
					
					Otherwise (i.e. if $\alpha^\omega \not\models \varphi_2$), then $\alpha^\omega \models \neg\varphi_1$. In that case, by Lemma~\ref{lem:distinguish}, it follows that $(\alpha_{\I_1})^\omega \models \neg \varphi_1$. with $\alpha_{\I_1} \subseteq X_{\I_1}$. Hence, $\alpha_{\I_1} \models \varphi^d_{\neg 1}$ and $\alpha_{\I_1} \models \varphi^d$. Therefore $\alpha \models \varphi^d$ since $\varphi^d$ is a disjunction. 
					
					Hence, $\varphi$ implies a disjunction on $\I$.
				
					\item If $\varphi_2$ accepts $\mathcal{N}$, this is similar to the previous item, by reversing the roles of $\varphi_1$ and $\varphi_2$.
				\end{enumerate}
				\item Assume that $\varphi$ rejects $\mathcal{P}_\I$ and accepts $\mathcal{N}$. 
				\begin{enumerate}
					\item Assume that $\varphi_1$ accepts $\mathcal{P}_{\I_1}$. Then, $\varphi_2$ rejects $\mathcal{P}_{\I_1}$, hence, $\varphi_2$ rejects $\mathcal{N}$, hence $\varphi_1$ rejects $\mathcal{N}$, hence $\varphi_1$ rejects $\mathcal{P}_{\I_2}$, hence $\varphi_2$ accepts $\mathcal{P}_{\I_2}$.

					Therefore $\varphi_1$ (resp. $\varphi_2$) accepts $\mathcal{P}_{\I_1}$ (resp. $\mathcal{P}_{\I_2}$) and rejects $\mathcal{N}$. By our induction hypothesis, $\varphi_1$ (resp. $\varphi_2$) implies a disjunction $\varphi^d_1$ on $\I_1$ (resp. $\varphi^d_2$ on $\I_2$). Let $\varphi^d_{\neg} := \varphi^d_1 \lor \varphi^d_2$ be a valuation on $\I$. Let us show that $\neg \varphi$ implies it. 
					
					Consider any $\alpha \subseteq X_{\I}$ and assume that $\alpha^\omega \models \neg \varphi$. If $\alpha^\omega \models \varphi_1$, then by Lemma~\ref{lem:distinguish}, it follows that $(\alpha_{\I_1})^\omega \models \varphi_1$, with $\alpha_{\I_1} \subseteq X_{\I_1}$. Hence, $\alpha_{\I_1} \models \varphi^d_1$ and $\alpha_{\I_1} \models \varphi^d$. Therefore $\alpha \models \varphi^d$ since $\varphi^d$ is a disjunction. 
					
					Otherwise (i.e. if $\alpha^\omega \not\models \varphi_1$), then $\alpha^\omega \models \varphi_2$. In that case, by Lemma~\ref{lem:distinguish}, it follows that $(\alpha_{\I_2})^\omega \models \varphi_2$, with $\alpha_{\I_2} \subseteq X_{\I_2}$. Hence, $\alpha_{\I_2} \models \varphi^d_2$ and $\alpha_{\I_2} \models \varphi^d$. Therefore $\alpha \models \varphi^d$ since $\varphi^d$ is a disjunction. 
					
					Hence, $\varphi$ implies a disjunction on $\I$.
					\item Assume that $\varphi_1$ rejects $\mathcal{P}_{\I_1}$. Then, $\varphi_2$ accept $\mathcal{P}_{\I_1}$, hence, $\varphi_2$ accepts $\mathcal{N}$, hence $\varphi_1$ accepts $\mathcal{N}$, hence $\varphi_1$ accepts $\mathcal{P}_{\I_2}$, hence $\varphi_2$ rejects $\mathcal{P}_{\I_2}$.
					
					Therefore $\varphi_1$ (resp. $\varphi_2$) rejects $\mathcal{P}_{\I_1}$ (resp. $\mathcal{P}_{\I_2}$) and accepts $\mathcal{N}$. By our induction hypothesis, $\neg \varphi_1$ (resp. $\neg \varphi_2$) defines a valuation $\varphi^d_{\neg 1}$ on $\I_1$ (resp. $\varphi^d_{\neg 1}$ on $\I_2$). Let $\varphi^d_{\neg} := \varphi^d_{\neg 1} \lor \varphi^d_{\neg 1}$ be a valuation on $\I$. Let us show that $\neg \varphi$ implies it. 
					
					Consider any $\alpha \subseteq X_{\I}$ and assume that $\alpha^\omega \models \neg \varphi$. If $\alpha^\omega \models \neg \varphi_1$, then by Lemma~\ref{lem:distinguish}, it follows that $(\alpha_{\I_1})^\omega \models \neg \varphi_1$. with $\alpha_{\I_1} \subseteq X_{\I_1}$. Hence, $\alpha_{\I_1} \models \varphi^d_{\neg 1}$ and $\alpha_{\I_1} \models \varphi^d_{\neg}$. Therefore $\alpha \models \varphi^d_{\neg}$ since $\varphi^d_{\neg}$ is a disjunction.
					
					Otherwise (i.e. if $\alpha^\omega \not\models \neg\varphi_1$), then $\alpha^\omega \models \neg\varphi_2$. In that case, by Lemma~\ref{lem:distinguish}, it follows that $(\alpha_{\I_2})^\omega \models \neg \varphi_2$. with $\alpha_{\I_2} \subseteq X_{\I_2}$. Hence, $\alpha_{\I_2} \models \varphi^d_{\neg 2}$ and $\alpha_{\I_2} \models \varphi^d_{\neg}$. Therefore $\alpha \models \varphi^d_{\neg}$ since $\varphi^d_{\neg}$ is a disjunction.
					
					Hence, $\varphi$ implies a disjunction on $\I$.
				\end{enumerate}
			\end{enumerate}
		\end{itemize}
		Overall, in any case, $\varphi$ implies a disjunction on $\I$.
	\end{proof}

	\subsubsection{Proof of Theorem~\ref{thm}}
	We can now proceed to the proof of Lemma~\ref{lem:hard}.
	\begin{proof}
		Consider an $\LTL$-formula $\varphi$ of size at most $B = 2m-1$ that accepts all formulae in $\mathcal{P}$ and rejects all formulae in $\mathcal{N}$. Since all the infinite words in $\mathcal{P}$ and $\mathcal{N}$ are size-1 infinite words, it follows, by  Lemma~\ref{lem:no_temp}, that this formula $\varphi$ can be chosen temporal-free, up to applying the function $\msf{tr}$. Furthermore, by Corollary~\ref{coro:at_least_k}, for all $k \in \Idm$, a variable in $X_k$ must occur in $\varphi$, i.e. $\prop(\varphi) \cap X_k \neq \emptyset$. Hence, since $\msf{Sz}(\varphi) \leq 2m-1$, it follows, by Lemma~\ref{lem:concise}, that $\prop(\varphi)$ is concise and $|\prop(\varphi)| = m$ and $|\prop(\varphi) \cap X_k| = 1$ for all $k \in \Idm$. Thus, $\varphi$ is $\Idm$-concise. Then, Lemma~\ref{lem:implies_disjunction} gives that $\varphi$ implies a disjunction on $\Idm$. Consider such a disjunction $\varphi^d$ on $\Idm$ that $\varphi$ implies on $\Idm$. Clearly, $\varphi^d$ accepts all words in $\mathcal{P}$ since $\varphi$ does. Furthermore, since the only binary logical operator occurring in $\varphi^d$ is the $\lor$ operator, and since no proposition is negated in $\varphi^d$, it follows that $\varphi^d$ rejects $\eta$. Hence, we can apply Lemma~\ref{lem:proof_simple_case} to obtain that there is a valuation satisfying the formula $\Phi$.
	\end{proof}

	The proof of Theorem~\ref{thm} is now direct.
	\begin{proof}
		This is a direct consequence of Lemmas~\ref{lem:easy} and~\ref{lem:hard} and the fact that the reduction described in Definition~\ref{def:reduction} can be computed in polynomial time.
	\end{proof}

	\subsection{Discussion}
	We conclude this section by a discussion on what our proof shows, and what it does not. First, our proof shows that the decision problem is $\LTL_\msf{Learn}$ is $\msf{NP}$-hard. However, it also shows on the way that variants of this problem are $\msf{NP}$-hard. More precisely:
	\begin{itemize}
		\item Since our reduction only uses size-1 infinite words, our proof shows that the restriction $\LTL_\msf{Learn}^{=1}$ of the decision problem $\LTL_\msf{Learn}$ to size-1 infinite words is still $\msf{NP}$-hard.
		\item Further, since in the proof of Lemma~\ref{lem:easy}, the formula that we exhibit has size exactly $B$, it follows that the variant of $\LTL_\msf{Learn}$ where we look for an $\LTL$-formula of size exactly $B$, instead of being at most $B$, is still $\msf{NP}$-hard.
		\item Any variant of $\LTL_\msf{Learn}$ where only the set unary operators are changed is still $\msf{NP}$-hard.
		\item Any variant of $\LTL_\msf{Learn}$ with less binary operators, while keeping either the $\lor$ binary logical operator or the $\lW$ binary temporal operator (recall Lemma~\ref{lem:equiv_temp}) is still $\msf{NP}$-hard.
	\end{itemize}
	
	However, they are still many unanswered questions:
	\begin{itemize}
		\item We have only considered classical binary operators, hence we do not handle arbitrary binary operators (seen as functions $\{\top,\bot\}^2 \rightarrow \{\top,\bot\}$). We believe that this could be done, however, this would add several cases to the proof of Lemma~\ref{lem:implies_disjunction}.
		\item Assuming that we have achieved the previous item, we believe that we could then handle arbitrary binary temporal operators (seen as functions $(\{\top,\bot\}^\omega)^2 \rightarrow \{\top,\bot\}$), because we consider size-1 infinite words. Hence, we could extend Lemma~\ref{lem:equiv_temp} to handle such arbitrary binary temporal operators.
		\item Overall, we believe that our work can be slightly adapted and/or extended to show that any fragment of $\LTL_\msf{Learn}$ that still contains $\lor$ or $\lW$ (the latter being a \textquotedblleft temporal version \textquotedblleft{} of the former) as binary operators is $\msf{NP}$-hard, regardless of what other operators are added or removed. However, we do not know what happens in a variant of $\LTL_\msf{Learn}$ where the operator $\lor$ (and any of its temporal version) cannot be used anymore. We believe that we can still prove that it is $\msf{NP}$-hard assuming that the $\wedge$ operator can be used (or any of its temporal version, such as $\lM$, recall Lemma~\ref{lem:equiv_temp}). The reduction for the $\wedge$ case could be seen as the dual of the reduction for the $\lor$ case that we have presented in this paper.  However, we do not know how to handle other operators. For instance, is the decision problem $\LTL_\msf{Learn}$ still $\msf{NP}$-hard if the only binary operator that can be used is $\Leftrightarrow$? (Note that this question is asked from a purely theoretical point of view, whether or not such a restriction would make sense in practice is also important to consider.)
	\end{itemize}

	\section{$\CTL$ learning is $\msf{NP}$-hard}
	The goal of this section is to show the theorem below. 
	\begin{theorem}
		\label{thm_CTL}
		The decision problem $\CTL_\msf{Learn}$ is $\msf{NP}$-hard.
	\end{theorem}
	
	We proceed via a reduction from the decision $\LTL_\msf{Learn}^{= 1}$ that is the variant of $\LTL_\msf{Learn}$ where all ultimately periodic words are of size 1. To properly define this reduction, let us first define how to translate a size-1 ultimately periodic word into a single-state Kripke structure.
	\begin{definition}
		Consider a set of propositions $\prop$ and a size-1 ultimately periodic word $w  = \alpha^\omega \in (2^\prop)^\omega$ for some $\alpha \subseteq \prop$. We define the Kripke structure $M_w = (S,I,R,L)$ where:
		\begin{itemize}
			\item $S := \{ q \}$;
			\item $I := S = \{ q \}$;
			\item $R := \{ (q,q) \}$;
			\item $L(q) := \alpha \subseteq \prop$.
		\end{itemize}
	\end{definition}

	Before we define our reduction, we want to express what the above definition satisfies. To do so, we define below how to translate $\CTL$-formulae into $\LTL$-formulae.
	
	\begin{definition}
		Consider a set of propositions $\prop$. We let $\msf{f}_{\CTL-\LTL}: \CTL \rightarrow \LTL$ be such that it removes all $\exists$ or $\forall$ quantifiers from a $\CTL$-formula, while keeping all other operators. More formally:
		\begin{itemize}
			\item $\msf{f}_{\CTL-\LTL}(p) = p$ for all $p \in \prop$;
			\item $\msf{f}_{\CTL-\LTL}(\neg \varphi) = \neg \msf{f}_{\CTL-\LTL}(\varphi)$;
			\item $\msf{f}_{\CTL-\LTL}(\diamond (\bullet  \msf{f}_{\CTL-\LTL}(\varphi)) = \bullet  \msf{f}_{\CTL-\LTL}(\varphi))$ for all $\bullet \in \mathsf{Op}_{Un}^{\msf{temp}}$ and $\diamond \in \{\exists,\forall\}$;
			\item $\msf{f}_{\CTL-\LTL}(\varphi_1 \bullet \varphi_2) = \msf{f}_{\CTL-\LTL}(\varphi_1) \bullet \msf{f}_{\CTL-\LTL}( \varphi_2)$ for all $\bullet \in \mathsf{Op}_{Un}^{\msf{lg}}$;
			\item $\msf{f}_{\CTL-\LTL}(\diamond (\varphi_1 \bullet \varphi_2)) = \msf{f}_{\CTL-\LTL}(\varphi_1) \bullet  \msf{f}_{\CTL-\LTL}( \varphi_2)$ for all $\bullet \in \mathsf{Op}_{Un}^{\msf{temp}}$ and $\diamond \in \{\exists,\forall\}$.
		\end{itemize}
	\end{definition}
	
	A straightforward proof by induction shows that the above definition satisfies the lemma below:
	\begin{lemma}
		\label{lem:translate_ltl_ctl}
		Consider a set of propositions $\prop$ and any $\CTL$-formula $\varphi$. We have:
		\begin{itemize}
			\item $\msf{Sz}(\varphi) = \msf{Sz}(\msf{f}_{\CTL-\LTL}(\varphi))$;
			\item for all $\alpha \subseteq \prop$, we have $\alpha^\omega \models \msf{f}_{\CTL-\LTL}(\varphi)$ if and only if $M_{\alpha^\omega} \models \varphi$.
		\end{itemize}
	
		Furthermore, the function $\msf{f}_{\LTL-\CTL}: \LTL \rightarrow \CTL$ that adds an $\exists$ quantifier in an $\LTL$-formula before every temporal operator, and therefore that defines a $\CTL$-formula that way, ensures the following:
		\begin{itemize}
			\item for all $\LTL$-formulae $\varphi$, $\msf{Sz}(\varphi) = \msf{Sz}(\msf{f}_{\LTL-\CTL}(\varphi))$;
			\item for all $\LTL$-formulae $\varphi$, we have $\msf{f}_{\CTL-\LTL}(\msf{f}_{\LTL-\CTL}(\varphi)) = \varphi$.
		\end{itemize}
	\end{lemma}
	
	We can now define the reduction that we consider. 
	\begin{definition}
		\label{def:reduc_ctl}
		Consider an instance $\msf{In}_{\LTL} = (\prop,\mathcal{P},\mathcal{N},B)$ of the $\LTL_\msf{Learn}^{=1}$ decision problem. We define the input $\msf{In}_{\CTL} = (\prop,\mathcal{P}',\mathcal{N}',B)$ with: 
		\begin{equation}
			\mathcal{P}' := \{ M_w \mid w \in \mathcal{P} \}
		\end{equation}
		and $\mathcal{N}'$ is defined similarly from $\mathcal{N}$.
	\end{definition}
	Clearly this reduction can be computed in polynomial time. Let us now show that it satisfies the desired property. 
	
	\begin{lemma}
		\label{lem:equiv_ctl}
		The input $\msf{In}_{\LTL}$ is a positive instance of the $\LTL_\msf{Learn}^{=1}$ decision problem if and only if $\msf{In}_{\CTL}$ is a positive instance of the $\CTL_\msf{Learn}$ decision problem.
	\end{lemma}
	\begin{proof}
		Assume that $\msf{In}_{\LTL}$ is a positive instance of the $\LTL_\msf{Learn}^{=1}$ decision problem. Let $\varphi_\LTL$ be an $\LTL$-formula separating $\mathcal{P}$ and $\mathcal{N}$ of size at most $B$. Consider the $\CTL$-formula $\varphi_\CTL := \msf{f}_{\msf{LTL}-\msf{CTL}}(\varphi_\LTL)$. By Lemma~\ref{lem:translate_ltl_ctl}, we have $\msf{Sz}(\varphi_\CTL) = \msf{Sz}(\varphi_\LTL) \leq B$ and  $\msf{f}_{\msf{CTL}-\msf{LTL}}(\varphi_\CTL) = \varphi_\LTL$.
		Consider now any $M_w \in \mathcal{P}'$ with $w \in \mathcal{P}$. Again by Lemma~\ref{lem:translate_ltl_ctl}, we have $M_{w} \models \varphi_\CTL$ if and only if $w \models \msf{f}_{\msf{CTL}-\msf{LTL}}(\varphi_\CTL) = \varphi_\LTL$, since $w$ is a size-1 ultimately periodic word. Since $w \in \mathcal{P}$, it follows that $w \models \varphi_\LTL$ and $M_{w} \models \varphi_\CTL$. This is similar for any $M_w \in \mathcal{N}'$ with $w \in \mathcal{N}$. Hence, $\msf{In}_{\CTL}$ is a positive instance of the $\CTL_\msf{Learn}$ decision problem. 
		
		Assume now that $\msf{In}_{\CTL}$ is a positive instance of the $\CTL_\msf{Learn}$ decision problem. Let $\varphi_\CTL$ be a $\CTL$-formula separating $\mathcal{P}'$ and $\mathcal{N}'$ of size at most $B$. Consider the $\LTL$-formula $\varphi_\LTL := \msf{f}_{\msf{CTL}-\msf{LTL}}(\varphi_\CTL)$. By Lemma~\ref{lem:translate_ltl_ctl}, we have $\msf{Sz}(\varphi_\LTL) = \msf{Sz}(\varphi_\CTL) \leq B$. Consider now any $w \in \mathcal{P}$. Again, by Lemma~\ref{lem:translate_ltl_ctl}, we have $M_{w} \models \varphi_\CTL$ if and only if $w \models \varphi_\LTL$, since $w$ is a size-1 ultimately periodic word. Since $M_{w} \in \mathcal{P}'$, it follows that $M_{w} \models \varphi_\CTL$ and $w \models \varphi_\LTL$. This is similar for any $w \in \mathcal{N}$. Hence, $\msf{In}_{\LTL}$ is a positive instance of the $\LTL_\msf{Learn}$ decision problem.
	\end{proof}
	
	The proof of Theorem~\ref{thm_CTL} is now direct.
	\begin{proof}
		It is straightforward consequence of Lemma~\ref{lem:equiv_ctl} and the fact that the reduction from Definition~\ref{def:reduc_ctl} can be computed in polynomial time.
	\end{proof}
	
	\section{Conclusion}
	We have proved that both $\LTL$ learning and $\CTL$ learning are $\msf{NP}$-hard. As argued in the introduction and stated in Proposition~\ref{prop:easy}, we obtain that both $\LTL$ learning and $\CTL$ learning are $\msf{NP}$-complete. Hence, it would seem that $\LTL$ learning and $\CTL$ learning are equivalently hard. However, we believe that, in some sense, $\CTL$ learning is harder than $\LTL$ learning. As a future work, we plan to investigate syntactical restrictions on formulae such that the resulting $\LTL$ fragment becomes easy (i.e. it could now be decided in polynomial time) whereas the resulting $\CTL$ fragment stays $\msf{NP}$-complete. 
	
	\bibliographystyle{splncs04}
	\bibliography{refs}
\end{document}